\title{\bf Pricing timer options and variance derivatives with closed-form partial transform under the 3/2 model}
\author{
        Wendong Zheng, \\
         Centre for Quantitative Finance,\\
         National University of Singapore, Singapore
         \vspace{2mm}\\
        Pingping Zeng,\footnote{Correspondence author; e-mail: pingpinghkust@gmail.com. Partially supported by the Austrian Science Fund (FWF) under
grant P25815 and the European Research Council (ERC) under grant FA506041.
}\\
        Department of Mathematics,\\
        University of Vienna, Austria\\
}
\date{}

\documentclass[12pt]{article}

\usepackage{amsmath}
\usepackage{epsfig,amsthm}
\usepackage{amssymb,amsfonts}
\usepackage{dsfont}
\usepackage{graphicx}
\usepackage{mathrsfs}
\usepackage[header,title,titletoc]{appendix}
\usepackage{float}
\usepackage{hyperref}
\usepackage{multirow}

\numberwithin{equation}{section}
\newtheorem{thm}{Theorem}
\newtheorem{corollary}{Corollary}
\newtheorem{remark}{Remark}
\newtheorem{lemma}{Lemma}
\newcommand{\D}{\mathrm{d}}
\newcommand{\I}{\mathrm{i}}
\newcommand{\E}{\mathbb{E}}

\newcommand{\bw}{\boldsymbol\omega}
\newcommand{\be}{\boldsymbol\eta}

\begin{document}
\maketitle

\begin{abstract}
\noindent
Most of the empirical studies on stochastic volatility dynamics favor the 3/2 specification over the square-root (CIR) process in the Heston model. In the context of option pricing, the 3/2 stochastic volatility model is reported to be able to capture the volatility skew evolution better than the Heston model. In this article, we make a thorough investigation on the analytic tractability of the 3/2 stochastic volatility model by proposing a closed-form formula for the partial transform of the triple joint transition density $(X,I,V)$ which stand for the log asset price, the quadratic variation (continuous realized variance) and the instantaneous variance, respectively. Two distinct formulations are provided for deriving the main result. The closed-form partial transform enables us to deduce a variety of marginal partial transforms and characteristic functions and plays a crucial role in pricing discretely sampled variance derivatives and exotic options that depend on both the asset price and quadratic variation. Various applications and numerical examples on pricing exotic derivatives with discrete monitoring feature are given to demonstrate the versatility of the partial transform under the 3/2 model.

\end{abstract}

\section{Introduction}

Stochastic volatility models (SVMs) were introduced to option pricing theory to resolve the incapability of the Black-Scholes framework in capturing the volatility smile/skew. Despite the difference in the specific assumption on the instantaneous volatility dynamics, the common practice of randomizing the volatility is to model the instantaneous volatility/variance as a correlated diffusion process. As pointed out by Itkin (2013), classic SVMs use a constant elasticity of variance (CEV) process for the instantaneous variance and there are just a few choices for the CEV parameter $\gamma$ that exhibit mathematical tractability. Nevertheless, various versions of SVMs have been proposed in the literature. Hull and While (1987) model the instantaneous variance process as a geometric Brownian motion ($\gamma=1$). Scott (1987) and Chesney and Scott (1989) let the log variance process be a mean-reverting Ornstein-Uhlenbeck (OU) process. Stein and Stein (1991) and Sch\"{o}bel and Zhu (1999) assume that the instantaneous volatility follows a mean-reverting OU process ($\gamma=0$). Heston (1993) instead proposes a mean-reverting square root process ($\gamma=1/2$) for the instantaneous variance. Barndorff-Nielsen and Shephard (2001) use a mean-reverting OU process to model the instantaneous variance process. Other variations of SVMs can be found in Bates (1996), Carr {\it et al.} (2003), Carr and Wu (2003), etc. Among all the proposed SVMs, the Heston model and its variants are the most popular ones. It is known that the Heston-type SVMs are essentially a subclass of the affine model family [see Duffie {\it et al.} (2000)].

Despite its popularity, there is a surprisingly large amount of empirical studies that report inconsistency of the affine models with market observations. Poteshman (1998) studies S\&P500 index option prices over a 7-year period and concludes that both the physical and risk-neutral drifts of the instantaneous variance are nonaffine. Also, the volatility of variance is observed to be an increasing convex function of the instantaneous variance. Using an affine drift CEV process for the instantaneous variance to fit the S\&P500 daily returns over a 30-year period, Ishida and Engle (2002) estimate the CEV parameter $\gamma=1.71$. In a similar work by Jones (2003), $\gamma$ is found to be 1.33 for daily S\&P100 returns and implied volatilities over a 14-year period. Chacko and Viceira (2003) employ the technique of the generalized method of moments (GMM) on a 35-year period of weekly returns and a 71-year period of monthly returns and estimate the CEV power to be 1.10 and 1.65, respectively. Javaheri (2004) tests three CEV power values: $\gamma=0.5,1,1.5$ on the time series of S\&500 daily returns and finds that $\gamma=1.5$ outperforms the other two. Using the same data as Jones (2003), Bakshi {\it et al.} (2006) test several SVMs on the time series of S\&P100 implied volatilities and find that a linear drift model is rejected and the coefficient of the quadratic term is highly significant. Moreover, their estimate of the CEV power is 1.27. The implication of these empirical findings is that one should use a diffusion process with nonaffine drift and CEV power greater than 1. With a quadratic drift and CEV power equal to 1.5, the 3/2 stochastic volatility model is obviously more firmly supported by the empirical study than the affine models.

Interestingly, the 3/2 model is not a brand-new model. In fact, it has already been examined by Heston (1997) and Lewis (2000) for its analytic tractability, applied to construct short rate models by Ahn and Gao (1999) and used to model credit default intensity by Andreasen (2001). However, until very recently the 3/2 model has not been considered as a mainstream SVM. The growing attention from the academia to the 3/2 model is partially attributed to the increasing interest in consistently modeling equity and volatility markets (represented by VIX), as a result of the booming of the volatility derivative market. In a recent work by Carr and Sun (2007), a new framework for option pricing that directly models the variance swap rate is proposed. They argue that the 3/2 specification for the instantaneous variance is a direct consequence of the model consistency requirement. Following Carr and Sun's work, Itkin and Carr (2010) introduce the 3/2 power time change process, and Chan and Platen (2010) price long dated variance swaps under the 3/2 stochastic volatility model. By performing extensive numerical comparisons between the Heston model and the 3/2 model, Drimus (2012) reports that the 3/2 model is superior to the Heston model in the sense that it is able to predict upward-sloping volatility of variance smiles, which is in good consistency with market observations. Most recently, Goard and Mazur (2013) report strong empirical evidence that VIX follows a 3/2 process other than an affine square root process. In an effort to consistently modeling VIX and equity derivatives, Baldeaux and Badran (2014) consider a 3/2 plus jumps model for pricing VIX derivatives. Yuen {\em et al.} (2014) derive closed form pricing formulas for exotic variance swaps under the 3/2 model.

Analytic tractability of an underlying model is essential for derivatives pricing. As for the 3/2 model, the characteristic function of the log asset price process is derived in Heston (1997) and Lewis (2000). In order to facilitate the pricing of volatility derivatives, Carr and Sun (2007) obtain the joint characteristic function of the log asset price and the quadratic variation, and Baldeaus and Badran (2012) extend their result to the 3/2 plus jumps model. Lewis (2014) derives the joint transition density of the log asset price and the instantaneous variance for the 3/2 model with constant parameters. Nevertheless, the full characterization of the 3/2 model seems to rely on the joint distribution of the triple $(X,I,V)$ which stand for the log asset price, quadratic variation and the instantaneous variance. Our paper fills this gap by providing a complete description of the joint distribution through the closed-form partial transform of the triple transition density. Two distinct formulations are proposed for deriving our main result. The partial differential equation (PDE) approach is inspired by the work by Carr and Sun (2007) and Lewis (2014). Carr and Sun (2007) solve the governing PDE with an ingenious choice of substituting variable and then reduce the PDE to an ODE. Lewis (2014) uses a sequence of transformations to reduce the original governing PDE to a first-order linear PDE which is solvable by the standard method of characteristics. Motivated by the joint exponential affine structure of the log asset price and quadratic variation, we define the partial transform of the triple joint density function and then solve the governing PDE by converting it to a Riccati system of ordinary differential equations through a Laplace transform. In the pure probabilistic formulation, the partial transform is factorized as a product of the conditional characteristic function of the integrated variance and the marginal transition density of the instantaneous variance. We then find the explicit expressions for both terms by using the change of measure technique and the reciprocal relation between the 3/2 process and CIR process.

Put in the relatively thin collection of literature on option pricing under the 3/2 model, the contribution of our work is three-fold. First, we propose a closed-form formula of the partial transform of the triple transition density under the 3/2 model with a time-dependent mean reversion parameter. The closed-form partial transform fully characterizes the analytic property of the 3/2 model, and most of the existing analytic formulas for the joint density functions and characteristic functions in the literature can be viewed as marginal versions of the newly derived partial transform. The closed-form partial transform enables us to accommodate the standard transform method to the pricing of derivative products whose terminal payoffs have exotic dependency on the asset price and quadratic variation under the 3/2 model with a time-dependent mean reversion parameter. Second, we provide two distinct derivations of the main result and each approach exhibits some unique novelty and improvement over the extant methodologies. In particular, our PDE approach takes advantage of the affine property of the pair $(X,I)$ and explicitly solves the governing PDE by converting the equation into a Riccati system of ODEs. Our probabilistic approach explores the relationship between the partial transform and the conditional characteristic function of the integrated variance and works out the closed-form expression for the latter with the change of measure technique. Third, using the newly derived formula, we manage to analytically price a variety of lately developed exotic derivative products, including finite-maturity discrete timer options and discretely sampled weighted moment swaps.

The rest of the paper is structured as follows. In Section 2, we provide a detailed description of the 3/2 model and its financial intuition. Section 3 is devoted to the derivation of our main result on the partial transform of the triple joint transition density under the PDE and probabilistic formulations. We also show that most marginal characteristic functions under the 3/2 model can be derived immediately from the partial transform of the triple. Section 4 covers examples for demonstrating the applications in derivatives pricing under the 3/2 model. In particular, the pricing of finite-maturity discrete timer options and discretely sampled weighted moment swaps are investigated in details. Numerical experiments and analyses are provided in Section 5 and conclusive remarks are given in Section 6.

\section{The 3/2 model}

Consider the 3/2 stochastic volatility model specified as follows:
\begin{equation}\label{themodel}
\begin{split}
  \frac{\D S_t}{S_t} &= (r-q)\,\D t + \sqrt{V_t}\big(\rho\,\D W^1_t + \sqrt{1-\rho^2}\,\D W^2_t\big),\\
  \D V_t &= V_t(\theta_t-\kappa V_t)\,\D t + \varepsilon V_t^{3/2}\,\D W^1_t,
\end{split}
\end{equation}
where $W^1_t$ and $W^2_t$ are two independent Brownian motions under the pricing measure $Q$. Here, we assume a constant riskfree rate $r$ and dividend yield $q$. Time-dependent deterministic riskfree rate and dividend yield can be accommodated without difficulty. In contrast to the square root process with affine drift in the Heston model, the parameters in the 3/2 variance dynamics need to be interpreted differently. The speed of mean reversion is now $\kappa V_t$, which is linear in $V_t$ and is a stochastic quantity. Since $\kappa>0$ under usual scenarios, the mean correction is quicker when the instantaneous variance is higher. Also, $\varepsilon$ cannot be interpreted as the same volatility of variance as in the Heston model. In fact, one needs to multiply it by a scaling factor $V_t$ in order to make it comparable to its counterpart in the Heston model. The long-term mean reversion level is given by $\theta_t/\kappa$. As pointed out by Itkin and Carr (2010), $\theta_t$ could be an independent stochastic process in the most general setting. By conditioning on the path of $\theta_t$, the analytic tractability of the 3/2 model with stochastic mean reversion level remains intact. Although in practice $\theta_t$ is often taken to be constant, this flexibility allows for more delicate modeling of the instantaneous variance dynamics when necessary.
Thus, in this paper, we assume $\theta_t$ to be a deterministic function of time, but an extension to an independent random process $\theta_t$ is not difficult.

Similar to the Heston model, not all choices of model parameters are admissible, in the sense that the non-explosion of $V_t$ and the martingale property of the discounted asset price process are guaranteed. According to Drimus (2012), the parameters of the 3/2 model specified by (\ref{themodel}) are constrained by
\begin{equation}
  \kappa - \rho\varepsilon \geq - \frac{\varepsilon^2}{2}.
\end{equation}
Notice that under normal market condition, $\kappa>0$ and $\rho\leq0$, the above inequality is automatically satisfied.

\section{The closed-form partial transform}

In this section, we present our main result on the closed-form formula for the partial transform defined previously. On top of deriving the main result by solving the governing PDE analytically, we provide an insightful probabilistic argument and some interesting intermediate results are obtained as a byproduct. We also briefly discuss how to deduce various marginal characteristic functions and transforms from the main result.

\subsection{The main result}

Let $X_t=\ln S_t$ be the log asset price and $I_t=\int_0^t V_s \,\D s$ be the quadratic variation of the log asset price process. Let $G(t,x,y,v;t',x',y',v')$ be the joint transition density of the triple $(X,I,V)$ from state $(x,y,v)$ at time $t$ to state $(x',y',v')$ at a subsequent time $t'$. By the Feynman-Kac Theorem, $G$ satisfies the following Kolmogorov backward equation:
\begin{equation}\label{pde}
\begin{split}
  -\frac{\partial G}{\partial t} &= \left(r-q-\frac{v}{2}\right)\frac{\partial G}{\partial x} + \frac{v}{2}\frac{\partial^2G}{\partial x^2} + v\frac{\partial G}{\partial y} +v(\theta_t-\kappa v)\frac{\partial G}{\partial v} + \frac{\varepsilon^2v^3}{2}\frac{\partial^2G}{\partial v^2}+\rho\varepsilon v^2\frac{\partial^2G}{\partial x\partial v},
\end{split}
\end{equation}
subject to the terminal condition:
\[
  G(t',x,y,v;t',x',y',v') = \delta(x-x')\delta(y-y')\delta(v-v'),
\]
where $\delta(\cdot)$ is the Dirac delta function.

The {\it partial transform} of $G$, denoted by $\check G$, is defined by
\begin{equation}\label{transformeddensity}
  \check G(t,x,y,v;t',\omega,\eta,v') = \int_{-\infty}^\infty\int_0^\infty e^{\I\omega x'+\I\eta y'}G(t,x,y,v;t',x',y',v')\,\D y'\D x'.
\end{equation}
The reason for introducing $\check G$ is that the affine structure of the pair $(X,I)$ makes $\check G$ more tractable than $G$ itself. Also, $\check G$ is more convenient to be used in option pricing. Obviously, $\check G$ also satisfies (\ref{pde}) with the terminal condition being $\check G(t',x,y,v;t',\omega,\eta,v')=e^{\I\omega x+\I\eta y}\delta(v-v')$. Since (\ref{pde}) has no coefficient involving $x$ or $y$, it is natural to come up with the following solution form:
\begin{equation}\label{checkG}
  \check G(t,x,y,v;t',\omega,\eta,v') = e^{\I\omega x+\I\eta y}g(t,v;t',\omega,\eta,v'),
\end{equation}
where $g$ satisfies the following PDE:
\begin{equation}\label{pde2}
  \begin{split}
  -\frac{\partial g}{\partial t} &=  \left[\I\omega\left(r-q-\frac{v}{2}\right)-\omega^2\frac{v}{2}+\I\eta v\right]g + [v(\theta_t-\kappa v)+\I\omega\rho\varepsilon v^2]\frac{\partial g}{\partial v}+\frac{\varepsilon^2v^3}{2}\frac{\partial^2g}{\partial v^2},
  \end{split}
\end{equation}
subject to the terminal condition:
\[
  g(t',v;t',\omega,\eta,v') = \delta(v-v').
\]
From (\ref{pde2}), we see that $g$ is essentially a function of $v$ and $t$, while all the other variables affect $g$ either through the PDE coefficients or boundary condition.
\begin{thm}\label{thm1}
Under the 3/2 stochastic volatility model specified by (\ref{themodel}), the partial transform of the triple joint transition density function defined by (\ref{transformeddensity}) is found to be
\begin{eqnarray}
  \check G(t,x,y,v;t',\omega,\eta,v') &=& e^{\I\omega x+\I\eta y}g(t,v;t',\omega,\eta,v'),\label{jointG}
\end{eqnarray}
where
\begin{equation}\label{fung}
  g(t,v;t',\omega,\eta,v') = e^{a(t'-t)}\frac{A_t}{C_t}\exp\left(-\frac{A_tv+v'}{C_tvv'}\right)\frac{1}{(v')^2}
   \left(\frac{A_tv}{ v'}\right)^{\frac{1}{2}+\frac{\tilde\kappa}{\varepsilon^2}}
  I_{2c}\left(\frac{2}{C_t}\sqrt{\frac{A_t}{vv'}}\right).
\end{equation}
Here,
\begin{eqnarray*}
  &&a = \I\omega(r-q),~~~\tilde\kappa =\kappa-\I\omega\rho \varepsilon,\\
 &&A_t = e^{\int_t^{t'}\theta_s\,\D s},~~~ C_t = \frac{\varepsilon^2}{2}\int_t^{t'}e^{\int_{t}^{s}\theta_{s'}\,\D s'}\,\D s,\\
&& c = \sqrt{\left(\frac{1}{2}+\frac{\tilde\kappa}{\varepsilon^2}\right)^2+\frac{\I\omega+\omega^2-2\I\eta}{\varepsilon^2}}.
\end{eqnarray*}
\end{thm}
\begin{proof}
We prove it by explicitly solving \eqref{pde2}. The details are given in Appendix~\ref{appA}.
\end{proof}
\begin{remark}
The solution form of $\check G$ given by \eqref{checkG} holds for any stochastic volatility process. In fact, if we consider a general stochastic volatility model specified by
\begin{equation}\label{heston}
\begin{split}
  \frac{\D S_t}{S_t} &= (r-q)\,\D t + \sqrt{V_t}\big(\rho\,\D W^1_t + \sqrt{1-\rho^2}\,\D W^2_t\big),\\
  \D V_t &= \alpha(t,V_t)\,\D t + \beta(t,V_t)\,\D W^1_t.
\end{split}
\end{equation}
It can be shown in a similar manner that the governing equation for $g$ is given by
\begin{equation}\label{hestonpde}
  \begin{split}
  -\frac{\partial g}{\partial t} &=  \left[\I\omega\left(r-q-\frac{v}{2}\right)-\omega^2\frac{v}{2}+\I\eta v\right]g + [\alpha(t,v)+\I\omega\rho \sqrt{v}\beta(t,v)]\frac{\partial g}{\partial v}+\frac{\beta(t,v)^2}{2}\frac{\partial^2g}{\partial v^2},
  \end{split}
\end{equation}
subject to the terminal condition:
\(
  g(t',v;t',\omega,\eta,v') = \delta(v-v').
\)
The partial transform of the model \eqref{heston} admits a closed-form expression if \eqref{hestonpde} can be solved explicitly.
\end{remark}
\begin{remark}
Theorem~\ref{thm1} can be easily extended to incorporate jumps in the asset price process. The analytic formula under the 3/2 plus jumps model is given in Appendix~\ref{3o2jump}.
\end{remark}

\subsubsection*{A probabilistic formulation of the partial transform}

The partial transform defined by \eqref{transformeddensity} admits an interesting conditional factorization, from which we can derive our main result in a pure probabilistic manner. Let $p_V(V_{t'}|V_t)$ be the transition density function of the instantaneous variance for any $t'>t$ under the pricing measure $Q$. Since $G(t,X_t,I_t,V_t;t',X_{t'},I_{t'},V_{t'})=G(t,X_t,I_t,V_t;t',X_{t'},I_{t'}|V_{t'})p_V(V_{t'}|V_t)$, we have
\begin{eqnarray}\label{G_condition}
  \check G(t,X_t,I_t,V_t;t',\omega,\eta,V_{t'})& =& \int_{-\infty}^\infty\int_0^\infty e^{\I\omega x'+\I\eta y'}G(t,X_t,I_t,V_t;t',x',y',V_{t'})\,\D y'\D x' \nonumber \\
  &=&\E_t [e^{\I \omega X_{t'}+\I \eta I_{t'}}|V_{t'}] p_V(V_{t'}|V_t),\label{relation1}
\end{eqnarray}
where $\E_t[\cdot]$ is short for $\E[\cdot|\mathcal{F}_t]$ under the measure $Q$ with $(\mathcal{F}_t)_{t\geq 0}$ being the natural filtration. The dynamics in (\ref{themodel}) can be rewritten to give [e.g. Zeng {\it et al.} (2013)]
\begin{equation*}
\begin{split}
 X_{t'}&=X_t+(r-q)(t'-t)+\frac{\rho}{\varepsilon} \left [\ln V_{t'}-\ln V_t-\int_t^{t'}\theta_s \,\D s \right]\\
 &\quad+\left [\rho\varepsilon \left(\frac{\kappa}{\varepsilon^2}+\frac{1}{2} \right)-\frac{1}{2}\right]\int_t^{t'} V_s\,\D s+\sqrt{1-\rho^2}\int_t^{t'}V_s \,\D W_s^2.
\end{split}
\end{equation*}
With this expression, the conditional expectation can be reformulated as
\begin{eqnarray}\label{joint_single}
  &&\E_t [e^{\I \omega X_{t'}+\I \eta I_{t'}}|V_{t'}] \nonumber\\
  &&=e^{\I \omega X_{t}+\I \eta I_{t}} \E_t\left[\E_t\Big[e^{\I \omega (X_{t'}-X_{t})}\Big|\int_t^{t'} V_s\,\D s,V_{t'}\Big] e^{\I \eta\int_t^{t'} V_s\,\D s}        \Bigg|V_{t'} \right] \nonumber \\
&&=e^{\I \omega X_{t}+\I \eta I_{t}}e^{\I \omega (r-q)(t'-t)}\left (\frac{V_{t'}}{A_tV_{t}} \right)^{\I \omega \frac{\rho}{\varepsilon}} \E_t \left[e^{ \left\{ \I \omega\left [ \rho\varepsilon \left(\frac{\kappa}{\varepsilon^2}+\frac{1}{2} \right)-\frac{1}{2}\right]-\frac{1}{2}(1-\rho^2)\omega^2+\I \eta \right\} \int_t^{t'} V_s\,\D s}\Big|V_{t'}\right].\label{relation2}
\end{eqnarray}
Therefore, the original problem boils down to the computation of the transition density of the variance process $p_V(V_{t'}|V_t)$ and the conditional characteristic function of the integrated variance in the form of $\E [e^{\I\xi \int_t^{t'} V_s\,\D s}|V_{t'},V_t]$.

\noindent
\textit{Transition density of the instantaneous variance}

\noindent
It is well known that the reciprocal of the 3/2 process is a CIR process. Write $U_t=\frac{1}{V_t}$, then $U_t$ is given by the following inhomogeneous CIR process:
\begin{equation}\label{U_tCIR}
  \D U_t=\left( (\kappa+\varepsilon^2)-\theta_t U_t \right)\D t-\varepsilon \sqrt{U_t}\,\D W_t^1.
\end{equation}
\begin{lemma}\label{thm_prob}
For any $t'>t$, $U_{t'}$ follows a non-central chi-square distribution conditional on $U_t$. Its transition (conditional) density function is given by
\begin{equation}\label{density_U}
  p_U(U_{t'}|U_{t})=\frac{A_t}{C_t} \mathrm{exp}\left(-\frac{A_t U_{t'}+U_t}{C_t}\right)
   \left(\frac{A_tU_{t'}}{U_t}\right)^{\frac{1}{2}+\frac{\kappa}{\varepsilon^2}}I_{1+\frac{2\kappa}{\varepsilon^2}}\left(\frac{2}{C_t}\sqrt{A_tU_{t'} U_t}\right),
\end{equation}
where $A_t$ and $C_t$ are given in Theorem~\ref{thm1}.
\end{lemma}
\begin{proof}
The inhomogeneous CIR process (\ref{U_tCIR}) equals a squared Bessel process $B(t)$ transformed by the following space-time changes in distribution
 \begin{equation*}
   (U_t, t\geq 0)=\left(\frac{1}{e^{\int_0^t \theta_s \,\D s}} B\left(\frac{\epsilon^2}{4}\int_0^t e^{\int_0^s \theta_{s'} \,\D s'} \,\D s \right), t\geq 0\right).
 \end{equation*}
 Based on the above relation and transition density of the squared Bessel process, Jeanblanc \emph{et al.} (2009) provide the transition density of the inhomogeneous CIR process. We refer interested readers to Jeanblanc \emph{et al.} (2009) for more details.
\end{proof}
It follows that the transition density of $V_t$ can be inferred to be
\begin{equation}\label{conditional_densityV}
  p_{V}(V_{t'}|V_t)=\frac{A_t}{C_t} \mathrm{exp}\left(-\frac{A_t V_t+V_{t'}}{C_tV_tV_{t'}}\right)
  \frac{1}{V^2_{t'}}\left(\frac{A_t V_t}{V_{t'}}\right)^{\frac{1}{2}+\frac{\kappa}{\varepsilon^2}}I_{1+\frac{2\kappa}{\varepsilon^2}}\left(\frac{2}{C_t}\sqrt{\frac{A_t}{V_{t'} V_t}}\right).
\end{equation}

\noindent
\textit{Conditional characteristic function of the integrated variance}

\noindent
Because of the strong path dependency of $\E[e^{\I\xi\int_t^{t'} V_s\,\D s}|V_{t'},V_t]$, the evaluation of this conditional expectation poses a major challenge.
Fortunately, with the change of probability measure technique, we manage to find the analytic formula for this conditional expectation.
\begin{thm}\label{thm2}
Under the 3/2 stochastic volatility model specified by (\ref{themodel}), the characteristic function of the integrated variance conditional on $V_{t'}$ admits
\begin{equation}\label{closed_conchf}
  \E \left[e^{\I\xi \int_t^{t'} V_s\,\D s}\big|V_{t'},V_t\right]=\frac{I_{\sqrt{\left(1+\frac{2\kappa}{\varepsilon^2} \right)^2-\frac{8\I\xi}{\varepsilon^2}}}\left(\frac{2}{C_t}\sqrt{\frac{A_t}{V_{t'} V_t}}\right)}{I_{1+\frac{2\kappa}{\varepsilon^2}}\left(\frac{2}{C_t}\sqrt{\frac{A_t}{V_{t'} V_t}}\right)},
\end{equation}
where the notations $A_t$ and $C_t$ are the same as those given in Theorem~\ref{thm1}.
\end{thm}
\begin{proof}
  The proof relies on a smart choice of a new probability measure under which the evaluation of expectations can be done conveniently. The details can be found in Appendix~\ref{appC}.
\end{proof}


\begin{remark}
Formula~\eqref{closed_conchf} can be viewed as a generalization of Baldeaux (2012) for a time-dependent mean reversion parameter $\theta_t$. Baldeaux's derivation is based on Lie symmetry and Laplace method, whereas in this article we manage to propose an easier alternative to obtain a more general result. Furthermore, one can now construct the exact simulation of the asset price process under the 3/2 model with a time-dependent mean reversion parameter.
\end{remark}

One can now easily verify that combining (\ref{G_condition}), (\ref{joint_single}), (\ref{conditional_densityV}) and (\ref{closed_conchf}) immediately leads to \eqref{jointG}.

\subsection{Marginal characteristic functions and transforms}

From the partial transform of the joint transition density of the triple $(X,I,V)$, we can deduce many useful marginal characteristic functions and transforms, including the renowned formulas proposed by Carr and Sun (2007) and Lewis (2014).

%
\begin{corollary}[Carr and Sun, 2007]\label{coro1}
Under the 3/2 stochastic volatility model specified by (\ref{themodel}), the joint characteristic function of the pair $(X,I)$ conditional on $\mathcal{F}_t$ is given by
\begin{equation}\label{jointchXI}
\begin{aligned}
  \E_t[e^{\I\omega X_{t'}+\I\eta I_{t'}}]
  &=e^{\I\omega X_t+\I\eta I_t} h(t,V_t;t',\omega,\eta),
  \end{aligned}
  \end{equation}
  where
  \begin{equation}\label{funh}
  h(t,v;t',\omega,\eta)=e^{a(t'-t)}\frac{\Gamma(\tilde\beta-\tilde\alpha)}{\Gamma(\tilde\beta)}
  \left(\frac{1}{C_tv}\right)^{\tilde\alpha}M\left(\tilde\alpha,\tilde\beta,-\frac{1}{C_tv}\right).
  \end{equation}
Here,
\begin{equation*}
  \tilde\alpha=-\frac{1}{2}-\frac{\tilde\kappa}{\varepsilon^2}+c,~~~\tilde\beta = 1+2c,
  \end{equation*}
  $\Gamma$ is the gamma function, and $M$ is the confluent hypergeometric function of the first kind.
\end{corollary}
\begin{proof}
See Appendix \ref{deriveh}.
\end{proof}
\begin{remark}
The marginal characteristic functions of $X$ and $I$ can be deduced by simply setting $\eta=0$ and $\omega=0$, respectively.
\end{remark}


Denote the joint transition density function of $(X,V)$ by $G_1(t,x,v;t',x',v')$ and define its partial transform by
\begin{equation}\label{checkGXV}
  \check G_1(t,x,v;t',\omega,v') = \int_{-\infty}^\infty e^{\I\omega x'}G_1(t,x,v;t',x',v')\,\D x'.
\end{equation}
\begin{corollary}\label{corol2}
Under the 3/2 stochastic volatility model specified by (\ref{themodel}), the partial transform defined by (\ref{checkGXV}) is given by
\begin{equation}\label{GXV}
  \check G_1(t,x,v;t',\omega,v') = e^{\I\omega x}g_1(t,v;t',\omega,v'),
\end{equation}
where
\begin{equation*}
\begin{aligned}
  g_1(t,v;t',\omega,v') &= g(t,v;t',\omega,0,v')\\
  &= e^{a(t'-t)}\frac{A_t}{C_t}\exp\left(-\frac{A_tv+v'}{C_tvv'}\right)(v')^{-2}
\left(\frac{A_tv}{ v'}\right)^{\frac{1}{2}+\frac{\tilde\kappa}{\varepsilon^2}}
  I_{2 c}\left(\frac{2}{C_t}\sqrt{\frac{A_t}{vv'}}\right)
  \end{aligned}
\end{equation*}
with
\[
   c = \sqrt{\left(\frac{1}{2}+\frac{\tilde\kappa}{\varepsilon^2}\right)^2+\frac{\I\omega+\omega^2}{\varepsilon^2}}.
\]
\end{corollary}
\begin{proof}
  By virtue of (\ref{transformeddensity}), $\check G_1$ can be obtained by simply setting $\eta=0$ in (\ref{jointG}).
\end{proof}
\begin{remark}
One can immediately see that this result can be viewed as a generalization of Lewis (2014) for a time-dependent mean reversion parameter $\theta_t$.
\end{remark}
\begin{remark}
  By further setting $\omega=0$ in (\ref{GXV}), we can retrieve the transition density of $V$:
  \begin{equation}\label{GV}
G_V(t,v;t',v')= g(t,v;t',0,0,v').
  \end{equation}
This hence provides an alternative derivation of (\ref{conditional_densityV}).
\end{remark}


The joint characteristic function of $(X,I)$ given by (\ref{jointchXI}) is useful for pricing European contingent claims that only depend on the values of $(X_T,I_T)$ at the maturity $T$. However, for more exotic derivatives, (\ref{jointchXI}) alone is not enough. For instance, the pricing of discretely monitored finite-maturity timer options also requires the joint characteristic function of $(X_{t_i},I_{t_j})$, where $i\neq j$. The evaluation of discretely sampled weight moment swaps needs the joint characteristic function of $(X_{t_i},X_{t_j},X_{t_k})$, where $i,j$ and $k$ are all different. These multivariate joint characteristic functions cannot be deduced from Corollary~\ref{coro1} or \ref{corol2}, and one needs to make use of the full characterization of the triple.

Let $\bw=(\omega_1,\cdots,\omega_m)^T$ and $\be =(\eta_1,\cdots,\eta_m)^T$ be scalar vectors in $\mathbb{R}^m$, and $\mathbf{X}_{\mathcal{T}}=(X_{t_1},\cdots,X_{t_m})^T$ and $\mathbf{I}_{\mathcal{T}}=(I_{t_1},\cdots,I_{t_m})^T$ be vectors of log asset price and quadratic variation, respectively, observed at a sequence of ascending future instants denoted by $\mathcal{T}=\{t_i\}_{i=1}^m$. Then, we define the multivariate joint characteristic function by
\begin{equation}\label{mulchf}
  \Phi(t,X_t,I_t,V_t;{\mathcal{T}},\bw,\be) = \E_t[e^{\I\bw\cdot\mathbf{X}_{\mathcal{T}}+\I\be\cdot\mathbf{I}_{\mathcal{T}}}],
\end{equation}
where $t$ denotes the current time and $t<t_i$ for any $i$. Specifically, the bivariate joint characteristic function with $m=2$ is of particular interest and will be used repeatedly in the subsequent section.
\begin{corollary}\label{coro2}
  Under the 3/2 stochastic volatility model specified by (\ref{themodel}), the bivariate joint characteristic function defined by (\ref{mulchf}) with $m=2$ is given by
  \begin{equation}\label{bichf}
  \begin{split}
    &\quad\Phi(t,X_t,I_t,V_t;\mathcal{T},\bw,\be) \\
     &= e^{\I(\omega_1+\omega_2)X_t+\I(\eta_1+\eta_2)I_t}
    \int_0^\infty g(t,V_t;t_1,\omega_1+\omega_2,\eta_1+\eta_2,v')h(t_1,v';t_2,\omega_2,\eta_2)\,\D v',
  \end{split}
  \end{equation}
  where $g$ and $h$ are given by (\ref{fung}) and (\ref{funh}), respectively.
\end{corollary}
\begin{proof}
  By direct calculation, we have
  \begin{eqnarray*}
    &\quad&  \Phi(t,X_t,I_t,V_t;\mathcal{T},\bw,\be)\\
    &=&\E_t[e^{\I\omega_1X_{t_1}+\I\eta_1I_{t_1}}\E_{t_1}[e^{\I\omega_2X_{t_2}+\I\eta_2I_{t_2}}]]\\
    &=& \E_t[e^{\I(\omega_1+\omega_2)X_{t_1}+\I(\eta_1+\eta_2)I_{t_1}}h(t_1,V_{t_1};t_2,\omega_2,\eta_2)] \quad\text{(by Corollary~\ref{coro1}) }\\
    &=& \int_0^\infty h(t_1,v';t_2,\omega_2,\eta_2) \int_{-\infty}^\infty\int_0^\infty
    e^{\I(\omega_1+\omega_2)x'+\I(\eta_1+\eta_2)y'}
    G(t,X_t,I_t,V_t;t_1,x',y',v')\,\D x'\D y'\D v'\\
    &=& \int_0^\infty h(t_1,v';t_2,\omega_2,\eta_2)\,\check G(t,X_t,I_t,V_t;t_1,\omega_1+\omega_2,\eta_1+\eta_2,v')\,\D v'\\
    &=& e^{\I(\omega_1+\omega_2)X_t+\I(\eta_1+\eta_2)I_t}
    \int_0^\infty g(t,V_t;t_1,\omega_1+\omega_2,\eta_1+\eta_2,v')h(t_1,v';t_2,\omega_2,\eta_2)\,\D v'.
  \end{eqnarray*}
  This completes the proof.
\end{proof}
For notational convenience, in the sequel, we suppress the dependence of $\Phi$ on $(X_t,I_t,V_t)$ and write $\Phi_t(\mathcal{T},\bw,\be) $ for short when there is no ambiguity.

\section{Applications in derivatives pricing}

In this section, we demonstrate how the partial transform of the triple and its induced marginal characteristic functions and transforms can be used to price various exotic derivative products that may be embedded with path-dependent feature, and have a complicated terminal payoff structure that depends on the asset price and quadratic variation.

\subsection{Finite-maturity discrete timer options}

The first presence of timer option in the literature dates back to the 90s when Neuberger (1990) and Bick (1995) discuss the pricing and hedging of such an imaginary product. Since the official launch of timer option by Soci\'{e}t\'{e} G\'{e}n\'{e}rale in 2008 [see Sawyer (2007)], research work that examines this exotic product from either an analytic or approximation perspective has been extensive [see Zeng {\it et al.} (2013) and references therein]. Voluminous as the literature is, most of the existing results are achieved under the assumption of perpetuity and continuous monitoring. Direct investigation on the discrete timer option with finite maturity is relatively rare and is more mathematically challenging. In this subsection, we will discuss how to analytically price discrete timer options with finite maturity under the 3/2 model based on our newly derived partial transform of the triple density.

Instead of having a deterministic expiry date as in a vanilla option, a discrete timer option with finite maturity expires on a random date which is either the first time when a pre-specified variance budget is fully consumed by the realized variance of the underlying asset price or the pre-specified mandatory expiry date, depending on which one comes earlier. Let $T$ be the mandatory expiry of the timer option and denote the tenor of the monitoring dates for the realized variance by $\{t_0,t_1\cdots,t_N\}$. For brevity, we assume equally spaced monitoring interval, i.e. $t_j=j\Delta=j\frac{T}{N},$ for $j=1,2,\cdots, N$.

At the initiation of the timer option, the investor specifies an expected investment horizon $T_0$ and a target volatility $\sigma_0$ to define a variance budget
\begin{equation*}
  B=\sigma_0^2 T_0.
\end{equation*}
 Let $\tau_B$ be the first time in the tenor of monitoring dates at which the discrete realized variance exceeds the variance budget $B$, namely,
 \begin{equation*}
   \tau_B=\mathrm{min} \left\{j \Bigg| \sum_{i=1}^j \left( \ln \frac{S_{t_i}}{S_{t_{i-1}}} \right)^2 \geq B \right \} \Delta.
 \end{equation*}
 Let the current time be $t_0=0$, then the price of a finite-maturity discrete timer call option can be decomposed into two components, depending on whether the variance budget is consumed up before the mandatory maturity or not.
\begin{eqnarray}
C_0 &=& \E_0[e^{-r(T\wedge\tau_B)}\max(S_{T\wedge\tau_B}-K,0)]\nonumber\\
&=& \E_0[e^{-rT} \mathrm{max}(S_T-K,0)\mathbf{1}_{\{\tau_B > T\}}+e^{-r\tau_B} \mathrm{max}(S_{\tau_B}-K,0)\mathbf{1}_{\{\tau_B \leq T\}}],\label{timer1}
\end{eqnarray}
where $K$ is the strike price.

To explicitly calculate the timer option price, we use the quadratic variation $I_t$ as a proxy of the discrete realized variance for the monitoring of the first hitting time. That is, we redefine $\tau_B$ as follows:
 \begin{equation*}
   \tau_B=\mathrm{min} \left\{j \Bigg| I_{t_j} \geq B \right \}\Delta.
 \end{equation*}
With this simplification, the price of a finite-maturity discrete timer call option can be conveniently computed by further decomposing it into a sequence of timerlets as follows:
\begin{equation}\label{timer2}
\begin{aligned}
C_0&=\E[\,e^{-rT} \mathrm{max}(S_T-K,0)\mathbf{1}_{\{I_T< B\}}]   \\
&\quad+\E \left[\,\sum_{j=0}^{N-1} e^{-rt_{j+1}} \left( \mathrm{max}(S_{t_{j+1}}-K,0)\mathbf{1}_{\{I_{t_j}< B\}}-\mathrm{max}(S_{t_{j+1}}-K,0)\mathbf{1}_{\{I_{t_{j+1}}< B\}} \right) \right].
\end{aligned}
\end{equation}
The key observation is that the event $\{\tau_B> t\}$ is equivalent to $\{I_t< B\}$. That explains the equivalence between the first term of (\ref{timer1}) and (\ref{timer2}). Other terms can be deduced similarly by noting that
\[
  \{\tau_B \leq T\} = \bigcup_{i=0}^{N-1} \{t_{i}< \tau_B \leq t_{i+1}\}=\bigcup_{i=0}^{N-1}\{I_{t_j}< B,I_{t_{j+1}}\geq B\}.
\]
The above decomposition is first proposed by Lee (2013) and has been applied in Cui (2014).

The series of expectations in \eqref{timer2} can be easily evaluated by the standard transform method [e.g. Lee (2013)], if we know the explicit expressions of the joint characteristic functions of $(X_{t_j}, I_{t_j})$ and $(X_{t_{j+1}}, I_{t_j})$. A direct application of Corollary \ref{coro1} and Corollary~\ref{coro2} gives
\begin{equation*}
\begin{aligned}
  \E_0[e^{\I\omega X_{t_j}+\I\eta I_{t_j}}]
  &=e^{\I\omega X_0+\I\eta I_0} h(t_0,V_0;t_j,\omega,\eta),\\
  \E_0[e^{\I\omega X_{t_{j+1}}+\I\eta I_{t_j}}] &=e^{\I \omega X_0+\I\eta I_0}\int_{0}^{\infty} g(0,V_0;t_j,\omega,\eta,v')h(t_j,v';t_{j+1},\omega,0) \,\D v',
  \end{aligned}
  \end{equation*}
Note that the Fourier transform of $(S_{t_{j+1}}-K)^+\mathbf{1}_{\{I_{t_j}< B\}}$ and $(S_{t_{j+1}}-K)^+\mathbf{1}_{\{I_{t_{j+1}}< B\}}$ admit the same analytic representation
\begin{equation*}
  \hat F(\omega,\eta)=\int_{-\infty}^{\infty}\int_{-\infty}^{\infty}e^{-\I\omega x-\I \eta y}(e^{x}-K)^+\mathbf{1}_{\{y< B\}}\,\D x\D y=\frac{K^{1-\I\omega}e^{-\I\eta B}}{(\I\omega+\omega^2)\I \eta},
\end{equation*}
where $\omega_I<-1$ and $\eta_I>0$.
By Parseval's theorem, the finite-maturity discrete timer option price can be derived as follows
\begin{eqnarray*}
  C_0&=&\frac{1}{4\pi^2}\int_{-\infty}^{\infty}\int_{-\infty}^{\infty}e^{-rT}\hat F(\omega,\eta) \E_0[e^{\I\omega X_{t_{N}}+\I\eta I_{t_N}}]\,\D \omega_R \eta_R
  +\sum_{j=0}^{N-1}\frac{1}{4\pi^2}\int_{-\infty}^{\infty}\int_{-\infty}^{\infty}e^{-rt_{j+1}}\\
  &~~&\left( \hat F(\omega,\eta) \E_0[e^{\I\omega X_{t_{j+1}}+\I\eta I_{t_j}}]-\hat F(\omega,\eta)  \E_0[e^{\I\omega X_{t_{j+1}}+\I\eta I_{t_{j+1}}}] \right) \,\D \omega_R \eta_R  \nonumber\\
  &=&\frac{1}{4\pi^2}\int_{-\infty}^{\infty}\int_{-\infty}^{\infty} \hat F(\omega,\eta)H(\omega,\eta)\,\D \omega_R \eta_R,
\end{eqnarray*}
where
\begin{equation*}
  \begin{aligned}
H(\omega,\eta)
=~&e^{-rT}e^{i\omega X_0+i\eta I_0} h(0,V_0;t_N,\omega,\eta)+ e^{\I \omega X_0+\I\eta I_0}\sum_{j=0}^{N-1}e^{-rt_{j+1}} \\
&\left(\int_{0}^{\infty} g(0,V_0;t_j,\omega,\eta,v')h(t_j,v';t_{j+1},\omega,0)\,\D v'- h(0,V_0;t_{j+1},\omega,\eta)  \right).
\end{aligned}
\end{equation*}

\subsection{Discretely sampled weighted moment swaps}

Let $0=t_0<t_1<\cdots<t_N=T$ be $N+1$ sampling dates and $\{i_k\}_{k=1}^N$ be a sequence of integers that take value from the index set $\{0,1,2,\cdots,N\}$. We define the $k$th weight $w_k$ as a mapping: $S_{t_{i_k}}\mapsto f(S_{t_{i_k}})$ for some function $f$. Then, the floating leg of a weighted $m$th moment swap is defined by
\[
  \frac{1}{T}\sum_{k=1}^Nf(S_{t_{i_k}})\left(\ln\frac{S_{t_k}}{S_{t_{k-1}}}\right)^m.
\]
Immediately, we identify the following special cases of discretely sampled weighted moment swaps.
\begin{table}[H]
  \begin{center}
    \begin{tabular}{llll}
      \hline
      product & $i_k$ & $f(x)$ & $m$\\ \hline
      variance swap & -- & $x\mapsto1$ & 2 \\
      gamma swap & $i_k=k$ & $x\mapsto x/S_0$ & 2 \\
      corridor variance swap & $i_k=k$ or $i_k=k-1$ & $x\mapsto 1_{\{l<x\leq u\}}$ & 2\\
      self-quantoed swap & $i_k=N$ & $x\mapsto x/S_0$ & 2\\
      skewness swap & -- & $x\mapsto1$ & 3\\
      \hline
    \end{tabular}
    \caption{Various weighted moment swaps}\label{tablewms}
  \end{center}
\end{table}
The pricing of the above weighted moment swap requires the computation of the fair strike price $K$ which is set to be
\[
  K = \frac{1}{T}\E_0\left[\sum_{k=1}^Nf(S_{t_{i_k}})\left(\ln\frac{S_{t_k}}{S_{t_{k-1}}}\right)^m\right],
\]
such that it costs zero for both parties to enter into the swap contract. Due to additivity of expectation, we have $K=\frac{1}{T}\sum_{k=1}^NL_k$, where
\[
  L_k = \E_0\left[f(S_{t_{i_k}})\left(\ln\frac{S_{t_k}}{S_{t_{k-1}}}\right)^m\right].
\]
To compute $L_k$, let us first consider the trivial case where $i_k=0$ or $f(x)\equiv 1$, meaning that the weight is deterministic. As a result, we have
\begin{eqnarray*}
  L_k &=& w_k\E_0\left[\left(\ln\frac{S_{t_k}}{S_{t_{k-1}}}\right)^m\right] \\
  &=& w_k\E_0\left[\frac{1}{\I^m}\frac{\partial^m}{\partial \phi^m}e^{\I\phi(X_{t_k}-X_{t_{k-1}})}\Big|_{\phi=0}\right]\\
  &=& \frac{w_k}{\I^m} \frac{\partial^m}{\partial \phi^m}\E_0[e^{\I\phi(X_{t_k}-X_{t_{k-1}})}]\Big|_{\phi=0} \\
  &=& \frac{w_k}{\I^m} \frac{\partial^m}{\partial \phi^m}\int_0^\infty G_V(t_0,V_0;t_{k-1},v')h(t_{k-1},v';t_k,\phi,0)\,\D v'\Big|_{\phi=0},
\end{eqnarray*}
The above calculation applies to variance swap and skewness swap. Note that in the last but second equality, the expectation is identified as the forward characteristic function first proposed by Hong (2004) and later used by Itkin and Carr (2010) to price discrete variance swaps.

Now, suppose $i_k\neq 0$ and $f$ is also nontrivial and assume that $f$ admits the generalized Fourier transform with respect to $\ln S_{t_{i_k}}$ as follows
\[
  \hat f(\omega) = \int_{-\infty}^\infty f(e^x)e^{-\I\omega x}\,\D x,
\]
where the transform variable $\omega$ is a complex number and its imaginary part $\omega_I$ is fixed in a way such that $|\hat f(\omega)|<\infty$. Then, we have
\begin{eqnarray}
  L_k &=& \E_0\left[\frac{1}{2\pi}\int_{-\infty}^\infty \hat f(\omega)e^{\I\omega X_{t_{i_k}}}\,\D \omega_R
  \frac{1}{\I^m}\frac{\partial^m}{\partial \phi^m}e^{\I\phi(X_{t_k}-X_{t_{k-1}})}\Big|_{\phi=0}\right] \nonumber\\
  &=& \frac{1}{2\pi\I^m}\int_{-\infty}^\infty \hat f(\omega) \frac{\partial^m}{\partial \phi^m}
  \E_0\left[e^{\I\omega X_{t_{i_k}}+\I\phi(X_{t_k}-X_{t_{k-1}})}\right]\Big|_{\phi=0}\,\D\omega_R.\label{Lk}
\end{eqnarray}
We remark that in \eqref{Lk}, $\E_0\left[e^{\I\omega X_{t_{i_k}}+\I\phi(X_{t_k}-X_{t_{k-1}})}\right]$ is a model-specific term, whereas all other terms (including the differential operator) can be considered as product-specific. This is a standard integral representation with the Fourier transform method, which holds for any pricing model.

In order to calculate the expectation in \eqref{Lk} under the 3/2 model, we first consider the scenario where $i_k\geq k$. Using the tower rule and the bivariate joint characteristic function (\ref{bichf}), we obtain
\begin{eqnarray}
  &\quad&\E_0\left[e^{\I\omega X_{t_{i_k}}+\I\phi(X_{t_k}-X_{t_{k-1}})}\right] \nonumber\\
  &=& \E_0\left[\E_{t_{k-1}}[e^{\I\phi X_{t_k}+\I\omega X_{t_{i_k}}}]e^{-\I\phi X_{t_{k-1}}}\right] \nonumber\\
  &=& \E_0\left[e^{\I(\omega+\phi)X_{t_{k-1}}-\I\phi X_{t_{k-1}}}
    \int_0^\infty h(t_k,v';t_{i_k},\omega,0)g_1(t_{k-1},V_{t_{k-1}};t_k,\omega+\phi,v')\,\D v'\right] \nonumber\\
  &=& e^{\I\omega X_0}\int_0^\infty \int_0^\infty h(t_k,v';t_{i_k},\omega,0)g_1(t_{k-1},v;t_k,\omega+\phi,v')g_1(t_0,V_0;t_{k-1},\omega,v)\,\D v'\D v.\label{ikbig}
\end{eqnarray}
Specifically, when $i_k=k$, the double integral can be further simplified as a single integral. Since $h(t_k,v';t_k,\omega,0)=1$ and $h(t_{k-1},v;t_k,\omega+\phi,0)=\int_0^\infty g_1(t_{k-1},v;t_k,\omega+\phi,v')\,\D v'$, we have
\[
  \E_0\left[e^{\I\omega X_{t_{k}}+\I\phi(X_{t_k}-X_{t_{k-1}})}\right]=
  e^{\I\omega X_0}\int_0^\infty h(t_{k-1},v;t_k,\omega+\phi,0)g_1(t_0,V_0;t_{k-1},\omega,v)\,\D v.
\]
In a similar manner, we obtain the following formula for the case $i_k\leq k-1$:
\begin{eqnarray}
  &\quad&\E_0\left[e^{\I\omega X_{t_{i_k}}+\I\phi(X_{t_k}-X_{t_{k-1}})}\right] \nonumber\\
   &=& e^{\I\omega X_0}\int_0^\infty \int_0^\infty h(t_{k-1},v';t_k,\phi,0)G_V(t_{i_k},v;t_{k-1},v')g_1(t_0,V_0;t_{i_k},\omega,v)\,\D v'\D v. \label{iksmall}
\end{eqnarray}
Further simplification is possible when $i_k=k-1$. In that case, $G_V(t_{k-1},v;t_{k-1},v')=\delta(v-v')$, and hence we have
\[
  \E_0\left[e^{\I\omega X_{t_{k-1}}+\I\phi(X_{t_k}-X_{t_{k-1}})}\right]=
  e^{\I\omega X_0}\int_0^\infty h(t_{k-1},v;t_k,\phi,0)g_1(t_0,V_0;t_{k-1},\omega,v)\,\D v.
\]
Observe that both (\ref{ikbig}) and (\ref{iksmall}) indicate that only one function in the integrand depends on the dummy variable $\phi$. As a result, the differentiation with respect to $\phi$ in (\ref{Lk}) can be conveniently performed on that single function. The final formula for $L_k$ is obtained as a triple integral (a double integral in the two special cases) by plugging (\ref{ikbig}) or (\ref{iksmall}) into (\ref{Lk}).

\noindent
{\it Self-quantoed variance swaps}

\noindent
As an illustrative example, we consider the fair strike price of a self-quantoed variance swap. According to Table~\ref{tablewms}, its fair strike is computed by
\[
  K = \frac{1}{T}\E_0\left[\sum_{k=1}^N\frac{S_T}{S_0}\left(\ln\frac{S_{t_k}}{S_{t_{k-1}}}\right)^2\right].
\]
Apparently, $L_k$ is now determined by (\ref{ikbig}). Moreover,
\[
  \hat f(\omega) = \frac{1}{S_0}\int_{-\infty}^\infty e^xe^{-\I\omega x}\,\D x =\frac{2\pi\delta(\omega+\I)}{S_0}.
\]
Consequently, the integration in (\ref{Lk}) can be computed explicitly as follows:
\begin{eqnarray*}
  L_k &=& -\frac{1}{S_0}\frac{\partial^2}{\partial \phi^2}
  \E_0\left[e^{ X_T+\I\phi(X_{t_k}-X_{t_{k-1}})}\right]\Big|_{\phi=0} \\
  &=& -\int_0^\infty \int_0^\infty h(t_k,v';T,-\I,0)\frac{\partial^2}{\partial \phi^2}g_1(t_{k-1},v;t_k,\phi-\I,v')\Big|_{\phi=0}g_1(t_0,V_0;t_{k-1},-\I,v)\,\D v'\D v.
\end{eqnarray*}
Therefore, we have
\begin{equation}
\begin{split}
  K &= -\frac{1}{T}\int_0^\infty \int_0^\infty \sum_{k=1}^N h(t_k,v';T,-\I,0)\\
  &~\quad\quad\frac{\partial^2}{\partial \phi^2}g_1(t_{k-1},v;t_k,\phi-\I,v')\Big|_{\phi=0}g_1(t_0,V_0;t_{k-1},-\I,v)\,\D v'\D v.
\end{split}
\end{equation}

\section{Numerical examples}
In this section, we show some numerical examples that illustrate the performance of pricing formulas of the timer option and variance derivatives under the 3/2 stochastic volatility model. We also examine the pricing behaviors with respect to varying model parameters.

Table \ref{3o2model} lists the default parameter values of the 3/2 model in our sample calculation. These parameter values, which have been calibrated to the S\&P 500 option prices, are taken from Drimus (2012). All our numerical examples will use this set of parameter values unless specified differently.
\begin{table}[H]
\begin{center}
  \begin{tabular}{cccccccc} \hline
$\kappa$&$\theta$&$\varepsilon$&$V_0$&$\rho$&$S_0$&$r$&$q$\\ \hline
22.84&4.979&8.56&0.060025&-0.99&100&0.015&0\\ \hline
  \end{tabular}
\end{center}
\caption{Parameter values in the 3/2 stochastic volatility model}\label{3o2model}
\end{table}

%

\subsection*{Finite-maturity discrete timer options}

We set the strike price to be $K =100$, and the number of monitoring instants to be $N=100$. We also reset the value of correlation coefficient $\rho=-0.5$ and initial variance $V_0=0.087$ which are adopted in Zeng {\it et al.} (2013) for comparison purposes. Considering that extensive analysis has already be done in Zeng {\it et al.} (2013) and references therein, here we mainly explore the impact of variance budget on finite-maturity discrete timer option call prices over a wide range of maturities.

Figure \ref{timer_options} shows that the finite-maturity discrete timer option price is an increasing function of both the maturity $T$ and the variance budget $B$. Intuitively, an option is usually more expensive with a longer exercise time. The higher variance budget $B$ or the longer the maturity $T$ leads to the later arrival of the exercise, thus giving a more expensive timer option price. When $T$ becomes sufficiently large, the finite-maturity discrete timer option value becomes almost insensitive to $T$ and converges to that of the perpetual discrete timer call option from below. The finite-maturity discrete timer option value exhibits a more pronounced convergence given a lower $B$. 

\begin{figure}[H]
  \centering
  \includegraphics[width=0.5\textwidth, bb = 3 5 489 455]{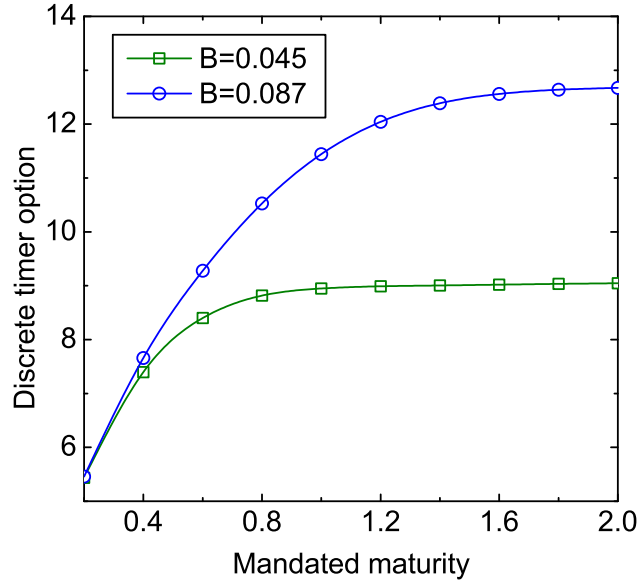}\\
  \caption{Plot of the finite-maturity discrete timer call option price versus mandated maturity given different values of variance budget.}\label{timer_options}
\end{figure}

\subsection*{Self-quantoed variance swaps}

\noindent We perform numerical analysis on the self-quantoed variance swap. Specifically, the fair strike prices of the self-quantoed variance swaps are compared with those of the variance swaps to exemplify the unique features of the former. In Figure~\ref{SQVfig}, two plots of the fair strike price versus the correlation coefficient $\rho$ and volatility of variance $\varepsilon$ are given for half-year daily sampled swap contracts ($N=126,\ \Delta t=1/252$).

In Figure~\ref{SQVfig} (a), we vary the volatility of variance parameter while fixing other parameters at the values given by Table~\ref{3o2model}, to analyze the sensitivity of the fair strike prices of the self-quantoed variance swap and the vanilla variance swap to $\varepsilon$. Within a realistic scope of $\varepsilon$, we observe large deviation in the fair strike prices of both swap products. Moreover, the fair strike prices are decreasing in $\varepsilon$. An intuitive explanation for this phenomenon can be found in Yuen {\em et al.} (2014). By contrasting the self-quantoed variance swap to the variance swap, we see that the disparity between the two gets wider as $\varepsilon$ decreases. That's because the protective feature of the self-quantoed variance swap begins to reinforce when the realized variance is getting large. A similar sensitivity analysis with respect to $\rho$ is performed in Figure~\ref{SQVfig} (b). As expected, the self-quantoed variance swap is much more dependent on the correlation coefficient than the variance swap. When $\rho<0$, we have the renowned leverage effect and the self-quantoed variance swap is embedded with a crash protection for the seller. On the other hand, when $\rho>0$, the soaring realized variance may be further exacerbated by the rising asset price, resulting in a larger risk exposure.

\begin{figure}[H]
  \centering
  \includegraphics[bb = 71 290 533 500]{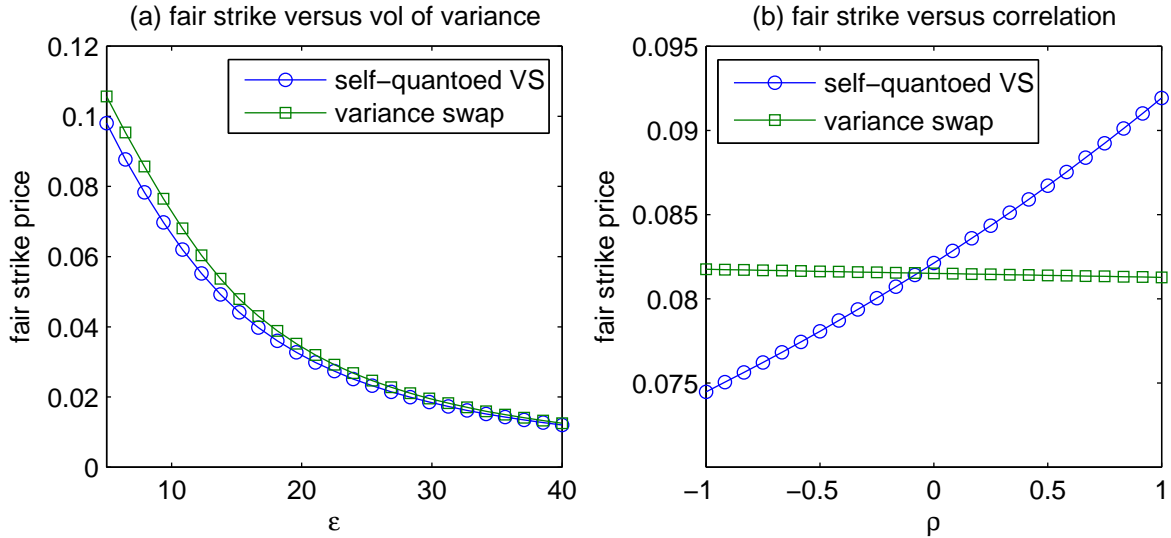}\\
  \caption{Comparison of the fair strike prices of a half-year daily sampled self-quantoed variance swap and a vanilla variance swap, with varying values of (a) vol of variance, (b) correlation coefficient.}\label{SQVfig}
\end{figure}

\section{Conclusions}

We explore the analytic tractability of the 3/2 stochastic volatility model by investigating the joint transition density of the triple consisting of the log asset price, quadratic variation and instantaneous variance. We obtain the closed-form partial transform of the joint triple density under the 3/2 stochastic volatility model with a time-dependent mean reversion parameter under two different formulations. In one approach, we establish the governing PDE of the partial transform which is then transformed to a Riccati system of ODEs and solved explicitly, while in the other one we relate the partial transform to the conditional characteristic function of the integrated variance and compute the latter one using a set of probabilistic tools, such as the change of measure technique. The newly derived partial transform serves as the underpinning of the transform methods for derivatives pricing under the 3/2 model and most extant characteristic functions or transforms can be viewed as marginal versions of the partial transform. As illustrative examples, the pricing formulas for the finite-maturity discrete timer options and self-quantoed variance swaps are expressed in terms of the partial transform and its induced marginal joint characteristic functions.

\subsection*{Acknowledgement}
We are deeply indebted to Professor Yue-Kuen Kwok\footnote{Department of Mathematics, Hong Kong University of Science and Technology} for his valuable advice on this research.

\newpage
\makeatletter
\bibliographystyle{abbrv}

\newpage
\begin{appendices}
\renewcommand{\theequation}{\Alph{section}.\arabic{equation}}

\section{Proof of Theorem~\ref{thm1}}\label{appA}

It suffices to solve (\ref{pde2}) which can be rewritten as
\begin{equation}\label{pde3}
  \frac{\partial \tilde g}{\partial t} + [\theta_tv-\tilde\kappa v^2]\frac{\partial \tilde g}{\partial v}
  +\frac{\varepsilon^2v^3}{2}\frac{\partial^2\tilde g}{\partial v^2} - \left[\frac{\I\omega+\omega^2}{2}-\I\eta\right]v\tilde g = 0,
\end{equation}
where $\tilde g=ge^{a(t-t')}$. The terminal condition remains to be $ \tilde g(t',v;t',\omega,\eta,v') = \delta(v-v').$
To solve (\ref{pde3}), we first change the variables
$(v,v')$ to $(u,u')$ with $u=1/v$ and $u'=1/v'$. We then define $f(t,u;t',\omega,\eta,u')$ that satisfies
\[
  \tilde g(t,v;t',\omega,\eta,v') = \frac{u^\alpha}{(u')^{\alpha-2}}f(t,u;t',\omega,\eta,u').
\]
By expressing the partial derivatives of $\tilde g$ in terms of the partial derivatives of $f$, we obtain the governing equation for $f$ as follows
\begin{eqnarray*}
   - \frac{\partial f}{\partial t} &=& \frac{\varepsilon^2u}{2}\frac{\partial^2f}{\partial u^2}+[\varepsilon^2(\alpha+1)+\tilde\kappa -\theta_t u]\frac{\partial f}{\partial u} -\alpha\theta_t f \\
   &\quad& + \left[\frac{\varepsilon^2}{2}(\alpha^2+\alpha)+\tilde\kappa\alpha-\frac{\I\omega+\omega^2}{2}+\I\eta\right]\frac{f}{u},
\end{eqnarray*}
subject to
\[
  f(t',u;t',\omega,\eta,u') = \delta(u-u').
\]
Apparently, if we choose $\alpha=\alpha(\omega,\eta)$ such that
\[
  \frac{\varepsilon^2}{2}(\alpha^2+\alpha)+\tilde\kappa\alpha-\frac{\I\omega+\omega^2}{2}+\I\eta = 0,
\]
then the governing PDE of $f$ becomes
\begin{equation}\label{affinepde1}
   - \frac{\partial f}{\partial t} = \frac{\varepsilon^2u}{2}\frac{\partial^2f}{\partial u^2}+[\varepsilon^2(\alpha+1)+\tilde\kappa -\theta_tu]\frac{\partial f}{\partial u} -\alpha\theta_t f,
\end{equation}
where all the coefficients are affine in $u$. It follows that $\alpha$ can take two values
\begin{equation}\label{alphasgn}
  \alpha= -\left(\frac{1}{2}+\frac{\tilde\kappa}{\varepsilon^2}\right)\pm c,
\end{equation}
where
\[
  c \equiv c(\omega,\eta) = \sqrt{\left(\frac{1}{2}+\frac{\tilde\kappa}{\varepsilon^2}\right)^2+\frac{\I\omega+\omega^2-2\I\eta}{\varepsilon^2}}.
\]
There are various ways of solving \eqref{affinepde1}. We find the following approach that makes use of the Riccati system of ODEs the best. An alternative method is appended to the end of this proof.

\noindent
\textit{Riccati system of ODEs}

\noindent
We write the Laplace transform of $f$ with respect to $u'$ as follows
\begin{equation}
 \hat f(t,u;t',\omega,\eta,\xi) = \int_0^\infty e^{-\xi u'}f(t,u;t',\omega,\eta,u')\,\D u'.
 \end{equation}
 Then, $\hat f$ satisfies
\begin{equation}\label{affinepde}
   - \frac{\partial\hat f}{\partial t} = \frac{\varepsilon^2u}{2}\frac{\partial^2\hat f}{\partial u^2}+[\varepsilon^2(\alpha+1)+\tilde\kappa -\theta_tu]\frac{\partial\hat f}{\partial u} -\alpha\theta_t\hat f,
\end{equation}
subject to the terminal condition $\hat f(t',u;t',\omega,\eta,\xi)=e^{-u\xi}$. By virtue of the affine structure of the PDE coefficients, (\ref{affinepde}) admits the following exponential affine solution:
\begin{equation}\label{affinesol}
  \hat f(t,u;t',\omega,\eta,\xi) = \exp(B(t,\xi)u+D(t,\xi)),
\end{equation}
where $B(t,\xi)$ and $D(t,\xi)$ are parameter functions determined by the following Riccati system of ODEs:
\begin{eqnarray*}
  -\frac{\partial B}{\partial t} &=& \frac{\varepsilon^2}{2}B^2 - \theta_t B,\\
  -\frac{\partial D}{\partial t} &=& [\varepsilon^2(\alpha+1)+\tilde\kappa]B-\alpha\theta_t,
\end{eqnarray*}
with boundary conditions $B(t',\xi)=-\xi$ and $D(t',\xi)=0$. It can be found that [e.g. Zheng and Kwok (2014)]
\begin{equation}\label{affineB}
  B(t,\xi) = -\frac{\xi}{A_t+C_t\xi},
\end{equation}
where
\begin{eqnarray*}
  A_t = e^{\int_t^{t'}\theta_s\,\D s},\quad C_t = \frac{\varepsilon^2}{2}\int_t^{t'}e^{\int_t^{s}\theta_{\tau}\,\D \tau}\,\D s.
\end{eqnarray*}
By noting that
\[
  \frac{\partial \ln(A_t+C_t\xi)}{\partial t} = -\theta_t -\frac{\varepsilon^2}{2}\frac{\xi}{A_t+C_t\xi},
\]
we obtain
\begin{equation}\label{affineD}
  D(t,\xi) = -2\left[\alpha+1+\frac{\tilde\kappa}{\varepsilon^2}\right]\ln(A_t+C_t\xi) +\left[\alpha+2+\frac{2\tilde\kappa}{\varepsilon^2}\right]\int_t^{t'}\theta_s\,\D s.
\end{equation}
Plugging (\ref{affineB}) and (\ref{affineD}) into (\ref{affinesol}) yields
\begin{equation}\label{hatf1}
  \hat f(t,u;t',\omega,\eta,\xi) = A_t^{\alpha+2+\frac{2\tilde\kappa}{\varepsilon^2}}\exp\left(-\frac{\xi u}{A_t+C_t\xi}\right)
  (A_t+C_t\xi)^{-2\alpha-2-\frac{2\tilde\kappa}{\varepsilon^2}}.
\end{equation}

\noindent
\textit{Inverse Laplace transform}

\noindent
Finally, $f$ can be retrieved by taking the inverse Laplace transform of $\hat f$ with respect to $\xi$. For $h(p)=p^{-\nu-1}e^{\gamma/p}$ for $\nu>-1$, the inverse Laplace transform of $h(p)$ is given by
\begin{equation}\label{laplaceinv}
  \mathcal{L}^{-1}[h](x) = \frac{1}{2\pi\I}\int_{\tau-\I\infty}^{\tau+\I\infty}e^{xp}h(p)\,\D p = \left(\frac{x}{\gamma}\right)^{\nu/2}I_{\nu}(2\sqrt{\gamma x}),
\end{equation}
where $I_\nu(\cdot)$ is the modified Bessel function of the first kind.
Closed-form expression of the inverse Laplace transform $\mathcal L^{-1} [\hat f]$ exists if we can ensure that
\begin{equation}\label{techcond1}
\mathfrak{R}\left(\alpha+1+\frac{\tilde\kappa}{\varepsilon^2}\right)>0,
\end{equation}
where $\mathfrak{R}(\cdot)$ means taking the real part. It will be shown later that this requires choosing the plus sign for $\alpha$ in (\ref{alphasgn}) so that
\begin{equation*}
\alpha=-\left(\frac{1}{2}+\frac{\tilde \kappa}{\epsilon^2}\right)+c.
\end{equation*}
Letting $p=A_t+C_t\xi$ and applying the inverse transform to $\hat f$, we have
\begin{eqnarray*}
  f(t,u;t',\omega,\eta,u') &=& \frac{A_t^{\frac{3}{2}+c+\frac{\tilde\kappa}{\varepsilon^2}}}{2\pi\I C_t}\int_{\tau-\I\infty}^{\tau+\I\infty}
  e^{\frac{u'(p-A_t)}{C_t}}p^{-2c-1}
  e^{-\frac{u(p-A_t)}{C_tp}}\,\D p\\
  &=& \frac{A_t^{\frac{3}{2}+c+\frac{\tilde\kappa}{\varepsilon^2}}e^{-\frac{u+A_tu'}{C_t}}}{2\pi\I C_t}\int_{\tau-\I\infty}^{\tau+\I\infty}
  e^{\frac{u'p}{C_t}}p^{-2c-1}
  e^{\frac{uA_t}{C_tp}}\,\D p\\
  &=& \frac{ A_t^{\frac{3}{2}-c+\frac{\tilde\kappa}{\varepsilon^2}}}{C_t}e^{-\frac{u+A_tu'}{C_t}}\left(\frac{A_tu'}{u}\right)
  ^{c}  I_{2c}\left(\frac{2}{C_t}\sqrt{A_tuu'}\right).
\end{eqnarray*}
Expressed in terms of $v$ and $v'$, $g$ is found to be
\begin{eqnarray*}
  g(t,v;t',\omega,\eta,v') = e^{a(t'-t)}\frac{A_t}{C_t}\exp\left(-\frac{A_tv+v'}{C_tvv'}\right)\frac{1}{(v')^2}\left(\frac{A_tv}{v'}\right)
  ^{\frac{1}{2}+\frac{\tilde\kappa}{\varepsilon^2}}
  I_{2c}\left(\frac{2}{C_t}\sqrt{\frac{A_t}{vv'}}\right).
\end{eqnarray*}
This completes the proof.

\noindent
\subsection*{An alternative method for solving (\ref{affinepde1})}\label{appB}

\noindent
Instead of performing the Laplace transform of $f$ with respect to $u'$ and taking advantage of the Riccati system of ODEs, one may solve (\ref{affinepde1}) using the method of characteristics. To proceed, we let $\tilde f=f\exp\left((\alpha-1)\int_t^{t'}\theta_s\,\D s\right)$, such that $\tilde f$ satisfies the following PDE:
\begin{equation}\label{eqf}
  \frac{\varepsilon^2}{2}\frac{\partial^2(u\tilde f)}{\partial u^2}-\frac{\partial((\theta_tu+\beta)\tilde f)}{\partial u} + \frac{\partial\tilde f}{\partial t} = 0,
\end{equation}
where $\beta=-\varepsilon^2\alpha-\tilde\kappa$. The terminal condition is given by
\[
  \tilde f(t',u;t',\omega,\eta,u') = \delta(u-u').
\]
By applying the Laplace transform to $\tilde f$ with respect to $u$:
$$\hat f(t,\xi;t',\omega,\eta,u')=\int_0^\infty e^{-\xi u}\tilde f(t,u;t',\omega,\eta,u')\,\D u,$$
the second order PDE (\ref{eqf}) is reduced to a first order linear PDE:
\begin{equation}\label{pde4}
  \frac{\partial \hat f}{\partial t} -\left(\varepsilon^2\xi^2/2-\theta_t\xi\right)\frac{\partial\hat f}{\partial \xi} - \beta\xi\hat f = 0,
\end{equation}
subject to
\[
  \hat f(t',\xi;t',\omega,\eta,u') = e^{-u'\xi}.
\]
\begin{lemma}
The solution to (\ref{pde4}) is given by
\begin{equation}\label{solpde4}
  \hat f(t,\xi;t',\omega,\eta,u') = \exp\left(-\frac{u'A_t\xi}{1+C_t\xi}\right)
  \left(\frac{1}{1+C_t\xi}\right)^{2\beta/\varepsilon^2},
\end{equation}
where
\begin{eqnarray*}
  A_t = e^{\int_t^{t'}\theta_s\,\D s}~~\text{and}~~ C_t = \frac{\varepsilon^2}{2}\int_t^{t'}e^{\int_t^{s}\theta_{s'}\,\D s'}\,\D s.
\end{eqnarray*}
\end{lemma}
\begin{proof}
Let $\zeta=\xi A_t$ and $\bar f(t,\zeta;t',\omega,\eta,u')=\hat f(t,\xi;t',\omega,\eta,u')$. Then, (\ref{pde4}) can be written as
\begin{equation}\label{pde5}
  \frac{\partial \bar f}{\partial t} - \frac{\varepsilon^2}{2}\frac{\zeta^2}{A_t}\frac{\partial \bar f}{\partial \zeta} = \beta\frac{\zeta}{A_t}\bar f,
\end{equation}
with terminal condition: $\bar f(t',\zeta;t',\omega,\eta,u')=e^{-u'\zeta}$.

Notice that (\ref{pde5}) is a first-order PDE and its characteristic equations are given by
\[
  \frac{\D\bar f}{A_t^{-1}\beta\zeta\bar f} = \frac{\D\zeta}{-A_t^{-1}\varepsilon^2\zeta^2/2} = \frac{\D t}{1}.
\]
If we define $\phi(t,\zeta,\bar f)=\ln\bar f + \frac{2\beta}{\varepsilon^2}\ln\zeta$ and $\psi(t,\zeta,\bar f)=\frac{1}{\zeta}+\frac{\varepsilon^2}{2}\int_t^{t'}A_s^{-1}\,\D s$, then according to the characteristic equations, we have
\begin{eqnarray*}
  \D\phi &=& \frac{\D\bar f}{\bar f} + \frac{2\beta}{\varepsilon^2}\frac{\D\zeta}{\zeta} = 0,\\
  \D\psi &=& -\frac{\D\zeta}{\zeta^2} - \frac{\varepsilon^2}{2}A_t^{-1}\D t = 0.
\end{eqnarray*}
As a result, the solution to (\ref{pde5}) must have the following form:
\[
  F(\phi,\psi) = F\left(\ln\bar f + \frac{2\beta}{\varepsilon^2}\ln\zeta,\ \frac{1}{\zeta}+\frac{\varepsilon^2}{2}\int_t^{t'}A_s^{-1}\,\D s\right) = 0,
\]
where $F$ is a function to be determined by the initial condition. Equivalently, we could rewrite the general solution in an explicit form of $\bar f$:
\[
  \ln\bar f = -\frac{2\beta}{\varepsilon^2}\ln\zeta + H\left(\frac{1}{\zeta}+\frac{\varepsilon^2}{2}\int_t^{t'}A_s^{-1}\,\D s\right),
\]
where $H$ is a function to be determined by the initial condition.

Using the initial condition: $\ln\bar f = -u'\zeta$, we obtain
\[
  -u'\zeta = -\frac{2\beta}{\varepsilon^2}\ln\zeta + H(\zeta^{-1}).
\]
As a result, $H$ is given by
\[
  H(z) = -\frac{u'}{z} - \frac{2\beta}{\varepsilon^2}\ln z.
\]
Therefore, the solution to (\ref{pde5}) is given by
\[
  \bar f(t,\zeta;t',\omega,\eta,u') = \exp\left(-\frac{u'\zeta}{1+\frac{\varepsilon^2}{2}\zeta\int_t^{t'}A_s^{-1}\,\D s}\right)
  \left(1+\frac{\varepsilon^2}{2}\zeta\int_t^{t'}A_s^{-1}\,\D s\right)^{-2\beta/\varepsilon^2}.
\]
By noting that $\zeta=A_t\xi$ and $C_t = \frac{\varepsilon^2}{2}A_t\int_t^{t'}A_s^{-1}\,\D s = \frac{\varepsilon^2}{2}\int_t^{t'}e^{\int_t^{s}\theta_{s'}\,\D s'}\,\D s$, we obtain (\ref{solpde4}).
\end{proof}
Finally, $\tilde f$ is obtained by taking inverse Laplace transform $\mathcal{L}^{-1}[\hat f]$ which can be expressed explicitly, provided that
\begin{equation}\label{techcond2}
\mathfrak{R}(\beta)>0.
\end{equation}
In this case, it requires choosing the minus sign for $\alpha$ in (\ref{alphasgn}), so that
\begin{equation}
\alpha=-\left(\frac{1}{2}+\frac{\tilde \kappa}{\epsilon^2}\right)-c.
\end{equation}
Using formula (\ref{laplaceinv}) and letting $p=1+C_t\xi$, we have
\begin{eqnarray*}
  \tilde f(t,u;t',\omega,\eta,u') &=& \frac{1}{2\pi\I}\int_{\tau-\I\infty}^{\tau+\I\infty}\frac{p^{-2c-1}}{C_t}
  \exp\left(\frac{u(p-1)-A_tu'(1-1/p)}{C_t}\right)\,\D p\\
  &=& \frac{\exp\left(-\frac{u+A_tu'}{C_t}\right)}{C_t}\frac{1}{2\pi\I}\int_{\tau-\I\infty}^{\tau+\I\infty}
  e^{\frac{u}{C_t}p}p^{-2c-1}e^{\frac{A_tu'}{C_tp}}\,\D p\\
  &=& \frac{\exp\left(-\frac{u+A_tu'}{C_t}\right) }{C_t}\left(\frac{u}{A_tu'}\right)^{c}
  I_{2c}\left(\frac{2}{C_t}\sqrt{A_tuu'}\right).
\end{eqnarray*}
Expressed in terms of the original variables $v$ and $v'$, $g$ is shown to have the same expression as (\ref{fung}).

\subsection*{The choice for value of $\alpha$}\label{appC}

When dealing with the inverse Laplace transform of $\hat f$ given by either (\ref{hatf1}) or (\ref{solpde4}), we need some technical condition which imposes constraints on the choice for value of $\alpha$. Here, we briefly discuss how different technical conditions lead to different choices for value of $\alpha$.

Denote $c=c_R+\I c_I$. Since
\[
  \mathfrak{R}\left(\left(\frac{1}{2}+\frac{\tilde\kappa}{\varepsilon^2}\right)^2+\frac{\I\omega+\omega^2-2\I\eta}{\varepsilon^2}\right)
  =\left(\frac{1}{2}+\frac{\kappa}{\varepsilon^2}\right)^2+\frac{\omega^2}{\varepsilon^2}(1-\rho^2)>0,
\]
we conclude that $c_R>0$ under the usual branching line convention.
Moreover,
\[
  c_R^2-c_I^2 = \mathfrak{R}\left(\left(\frac{1}{2}+\frac{\tilde\kappa}{\varepsilon^2}\right)^2+\frac{\I\omega+\omega^2-2\I\eta}{\varepsilon^2}\right)
  =\left(\frac{1}{2}+\frac{\kappa}{\varepsilon^2}\right)^2+\frac{\omega^2}{\varepsilon^2}(1-\rho^2),
\]
implying that
\[
  c_R^2 > \left(\frac{1}{2}+\frac{\kappa}{\varepsilon^2}\right)^2.
\]
In conclusion, we must have $c_R>\frac{1}{2}+\frac{\kappa}{\varepsilon^2}$.
Consequently, we have to choose $\alpha = -\left(\frac{1}{2}+\frac{\tilde\kappa}{\varepsilon^2}\right)+ c$ in order to meet (\ref{techcond1}) and $\alpha = -\left(\frac{1}{2}+\frac{\tilde\kappa}{\varepsilon^2}\right)- c$ in order to satisfy (\ref{techcond2}).

\section{The 3/2 plus jumps model}\label{3o2jump}

In this section, we extend our results to the 3/2 plus jumps model, which is given by
\begin{equation}\label{plusjump}
\begin{split}
  \frac{\D S_t}{S_t} &= (r-q-\lambda \vartheta)\,\D t + \sqrt{V_t}\big(\rho\,\D W^1_t + \sqrt{1-\rho^2}\,\D W^2_t\big)+\big(e^{J}-1\big)\,\D N_t,\\
  \D V_t &= V_t(\theta_t-\kappa V_t)\,\D t + \varepsilon V_t^{3/2}\,\D W^1_t,
\end{split}
\end{equation}
where a jump component modeled by a compound Poisson process is appended to the asset price process of the original 3/2 model (\ref{themodel}). Here, $N_t$ is an independent Poisson process with constant intensity $\lambda$ and $J$ stands for the random jump size which is assumed to be a normal random variable with mean $\mu$ and variance $\sigma^2$. Moreover, the additional term in the drift $\lambda\vartheta\D t$ is the compensator for the jump component, where
\[
  \vartheta = \E\left[e^{J}-1\right]=e^{\mu+\sigma^2/2}-1.
\]
With jumps in the asset price, the quadratic variation process is now given by $I_t=\int_0^t V_s\,\D s+\sum_{m=1}^{N_t}(J_m)^2$, and the triple transition density $G$ satisfies the following PIDE:
\begin{equation*}
\begin{split}
  -\frac{\partial G}{\partial t} &=  \left(r-q-\lambda\vartheta-\frac{v}{2}\right)\frac{\partial G}{\partial x} + \frac{v}{2}\frac{\partial^2G}{\partial x^2} + v\frac{\partial G}{\partial y} +v(\theta_t-\kappa v)\frac{\partial G}{\partial v} + \frac{\varepsilon^2v^3}{2}\frac{\partial^2G}{\partial v^2}+\rho\varepsilon v^2\frac{\partial^2G}{\partial x\partial v}\\
  &\quad+\lambda\int_{-\infty}^\infty (G(t,x+z,y+z^2,v;t',x',y',v')-G(t,x,y,v;t',x',y',v'))f_{J}(z)\,\D z,
\end{split}
\end{equation*}
where $f_{J}=\frac{1}{\sigma\sqrt{2\pi}}e^{-\frac{(x-\mu)^2}{2\sigma^2}}$ is the probability density of the jump size $J$. The partial transform defined by (\ref{transformeddensity}) admits a similar form given by
\[
  \check G(t,x,y,v;t',\omega,\eta,v') = e^{\I\omega x+\I\eta y}g(t,v;t',\omega,\eta,v'),
\]
where $g$ satisfies that following PDE:
\begin{equation*}
  \begin{split}
  -\frac{\partial g}{\partial t} &=  \left[\I\omega\left(r-q-\lambda\vartheta-\frac{v}{2}\right)-\omega^2\frac{v}{2}+\I\eta v\
  +\exp\left(\frac{2\I\mu(\omega+\eta\mu)-\omega^2\sigma^2}{2(1-2\I\eta\sigma^2)}\right)\frac{\lambda}{\sqrt{1-2\I\eta\sigma^2}}-\lambda\right]g \\
  &\quad + [v(\theta_t-\kappa v)+\I\omega\rho\varepsilon v^2]\frac{\partial g}{\partial v}+\frac{\varepsilon^2v^3}{2}\frac{\partial^2g}{\partial v^2}.
  \end{split}
\end{equation*}
The above PDE can be solved in the exactly same manner and the solution has the same expression as in (\ref{fung}) with a modified parameter $a$ given by
\[
  a = \I\omega(r-q-\lambda\vartheta)+\exp\left(\frac{2\I\mu(\omega+\eta\mu)-\omega^2\sigma^2}{2(1-2\I\eta\sigma^2)}\right)\frac{\lambda}{\sqrt{1-2\I\eta\sigma^2}}-\lambda.
\]

\section{Proof of Theorem~\ref{thm2}}\label{appC}

To determine the new measure under which the calculation can be done conveniently, we need the following lemma. Let $\mathcal{D}$ be an open subset and $Y_t\in \mathcal{D}\cup \partial\mathcal{D}$ be a one-dimensional stationary time-inhomogeneous diffusion process satisfying
\begin{equation*}
\D Y_t=\mu (Y_t, t)\,\D t+\sigma(Y_t)\,\D W_t,
\end{equation*}
where $W_t$ is a Brownian motion under the measure $Q$. Its infinitesimal generator is given by
\begin{equation*}
  \mathcal L^Q=\mu (y, t)\partial_y+\frac{1}{2}\sigma^2(y)\partial^2_y.
\end{equation*}

\begin{lemma}\label{change_measure}
Assume $\gamma(y)$ is twice differentiable and $\gamma(y)>0$ for any $y\in\mathcal{D}$. Define a new measure
$\tilde Q$ by the Radon-Nikodym derivative
\begin{equation}
\frac{\D\tilde Q}{\D Q}\Big|_{\mathcal F_t}=\frac{\gamma(Y_{t'})}{\gamma(Y_{t})}e^{-\int_t^{t'}\frac{\mathcal L^Q \gamma(Y_s)}{\gamma(Y_s)}\,\D s}.
\end{equation}
Suppose $Y_t$ remains to be finite and the boundary $\partial \mathcal D$ is unattainable under both $Q$ and $\tilde Q$, then $\tilde Q$ is equivalent to $Q$ and
\begin{equation}
\E_t \left[e^{-\int_t^{t'}\frac{\mathcal L^Q \gamma(Y_s)}{\gamma(Y_s)}\,\D s}l(Y_{t'}) \right]=\gamma(Y_t)\tilde \E_t \left[\frac{l(Y_{t'})}{\gamma(Y_{t'})}\right]
\end{equation}
for any bounded measurable function $l(y)$, where $\tilde\E_t[\cdot]$ is the expectation taken under $\tilde Q$. Moreover, under the new measure $\tilde Q$, the dynamics of $Y_t$ is modified to be
\begin{equation}
\D Y_t=\left(\mu (Y_t, t)+\sigma^2(Y_t)\frac{\gamma'(Y_t)}{\gamma(Y_t)}\right)\,\D t+\sigma(Y_t)\,\D \tilde W_t,
\end{equation}
where $\tilde W_t$ is a $\tilde Q$-Brownian motion.
\end{lemma}

Lemma~\ref{change_measure} is an extension of a related theorem established by Hurd and Kuznetsov (2008) from a general time-homogeneous diffusion process to a time-inhomogeneous diffusion process. It can be proved in a similar way based on the integrated It$\mathrm{\hat o}$ formula and the application of the Girsanov theorem. Equipped with Lemma~\ref{change_measure}, we are now able to obtain the conditional characteristic function of the integrated variance in closed form.

Note that $\E [e^{\I\xi \int_t^{t'} V_s\,\D s}\big|V_{t'},V_t]=\E_t [e^{\int_t^{t'}\frac{\I\xi}{U_s}\,\D s}|U_{t'}]$ and $\E_t[e^{\int_t^{t'}\frac{\I\xi }{U_s}\,\D s}e^{-\I\varpi U_{t'}}]$ can be evaluated in two different ways for any $\varpi$, where $U_t$ is a CIR process specified by (\ref{U_tCIR}). On one hand, we consider the calculation of $\E_t[e^{\int_t^{t'}\frac{\I\xi }{U_s}\,\D s}e^{-\I\varpi U_{t'}}]$ based on Lemma~\ref{change_measure}. We first need to find a twice differentiable function $\gamma(y)$ such that
\begin{equation}\label{Lgamma}
   \frac{\mathcal L^Q \gamma(U_s)}{\gamma(U_s)} = C-\frac{\I\xi}{U_s},
\end{equation}
for some constants $C$. It can be easily verified that $\gamma(y)=y^\nu$ is a solution. In fact, since $\mu (y,t)=(\kappa+\varepsilon^2)-\theta_t y$ and $\sigma(y)=-\varepsilon y$, we have
\begin{equation*}
\begin{aligned}
\frac{\mathcal L^Q \gamma(U_s)}{\gamma(U_s)}
&=-\theta_s \nu +\frac{(\kappa+\varepsilon^2)\nu+\frac{\varepsilon^2}{2}\nu(\nu-1)}{U_s}.
\end{aligned}
\end{equation*}
Apparently, (\ref{Lgamma}) holds if $C=-\theta_s \nu$ and $\nu$ is the solution to
\[
  \frac{\varepsilon^2}{2}\nu^2+(\kappa+\frac{\varepsilon^2}{2})\nu+{\I\xi} = 0,
\]
which admits
\begin{equation*}
\nu=-\frac{\kappa}{\varepsilon^2}-\frac{1}{2}\pm\sqrt{\left(\frac{1}{2}+\frac{\kappa}{\varepsilon^2} \right)^2-\frac{2\I\xi}{\varepsilon^2}}.
\end{equation*}
Applying Lemma~\ref{change_measure} to $U_t$ with $\gamma(y)=y^\nu$ and $l(y)=e^{-\I\varpi y}$ yields
\begin{eqnarray}
\E_t \left[e^{\int_t^{t'}\frac{\I\xi}{U_s}\,\D s}e^{-\I\varpi U_{t'}} \right]&=&e^{-\nu \int_t^{t'}\theta_s\,\D s}\E_t \left[e^{-\int_t^{t'}\frac{\mathcal L^Q \gamma(U_s)}{\gamma(U_s)}\,\D s}e^{-\I\varpi U_{t'}} \right]  \nonumber\\
&=&e^{-\nu \int_t^{t'}\theta_s\,\D s}U_t^\nu\tilde \E_t \left[U_{t'}^{-\nu}e^{-\I\varpi U_{t'}}\right] \nonumber\\
&=&e^{-\nu \int_t^{t'}\theta_s\,\D s}U_t^\nu \int_0^\infty e^{-\I\varpi U_{t'}} U_{t'}^{-\nu}\tilde p_U(U_{t'}|U_t)\,\D U_{t'},\label{prob_expec1}
\end{eqnarray}
where $\tilde p_U$ is the transition density of the CIR process $U_t$ under the new measure $\tilde Q$. Moreover, the $\tilde Q$-dynamics of $U_t$ is given by
\begin{equation}
  \D U_t=\left( (\kappa+\varepsilon^2+\varepsilon^2 \nu)-\theta_t U_t \right)\,\D t-\varepsilon \sqrt{U_t}\,\D \tilde W^1_t,
\end{equation}
where $\tilde W^1_t$ is a $\tilde Q$-Brownian motion. According to Lemma~\ref{thm_prob}, the transition density of $U_t$ under $\tilde Q$ is given by
\begin{equation}\label{density2_U}
  \tilde p_U(U_{t'}|U_{t})=\frac{A_t}{C_t} \mathrm{exp}\left(-\frac{A_t U_{t'}+U_t}{C_t}\right)
   \left(\frac{A_tU_{t'}}{U_t}\right)^{\frac{1}{2}+\frac{\kappa}{\varepsilon^2}+\nu}I_{1+\frac{2\kappa}{\varepsilon^2}+2\nu}\left(\frac{2}{C_t}\sqrt{A_tU_{t'} U_t}\right).
\end{equation}
Notice that in order for make zero an unattainable boundary of $U_t$ under $\tilde Q$, the Feller condition $\mathfrak{R}(\kappa+\varepsilon^2+\varepsilon^2 \nu)>\varepsilon^2/2$ must be observed. Then, it follows that
\begin{equation}\label{exp_nu}
\nu=-\frac{\kappa}{\varepsilon^2}-\frac{1}{2}+\sqrt{\left(\frac{1}{2}+\frac{\kappa}{\varepsilon^2} \right)^2-\frac{2\I\xi}{\varepsilon^2}}.
\end{equation}
On the other hand, we have the following alternative representation:
\begin{eqnarray}
\E_t \left[e^{\int_t^{t'}\frac{\I\xi}{U_s}\,\D s}e^{-\I\varpi U_{t'}} \right]
&=&\E_t \left[e^{-\I\varpi U_{t'}} \E_t \left[e^{\int_t^{t'}\frac{\I\xi}{U_s}\,\D s} \Big|U_{t'}\right] \right] \nonumber \\
&=&\int_0^\infty e^{-\I\varpi U_{t'}} \E_t \left[e^{\int_t^{t'}\frac{\I\xi}{U_s}\,\D s}\Big|U_{t'}\right] p_U(U_{t'}|U_t)\,\D U_{t'}.\label{prob_expec2}
\end{eqnarray}
By equating (\ref{prob_expec1}) and (\ref{prob_expec2}), we obtain
\begin{equation*}
e^{-\nu \int_t^{t'}\theta_s\,\D s}U_t^\nu \int_0^\infty e^{-\I\varpi U_{t'}} U_{t'}^{-\nu}\tilde p_U(U_{t'}|U_t)\,\D U_{t'}=\int_0^\infty e^{-\I\varpi U_{t'}} \E_t \left[e^{\int_t^{t'}\frac{\I\xi}{U_s}\,\D s}\Big|U_{t'}\right] p_U(U_{t'}|U_t)\,\D U_{t'}.
\end{equation*}
Noticing that the above equation holds for any $\varpi$, we can then take the inverse Fourier transform with respect to $\varpi$ on both sides to obtain
\begin{equation}\label{condi_mediate}
\E_t\left[e^{\int_t^{t'}\frac{\I\xi}{U_s}\,\D s}\Big|U_{t'}\right]=\frac{\left(\frac{U_t}{A_tU_{t'}}\right)^\nu \tilde p_U(U_{t'}|U_t)}{p_U(U_{t'}|U_t)}.
\end{equation}
Substituting (\ref{density_U}), (\ref{density2_U}) and (\ref{exp_nu}) into (\ref{condi_mediate}), we have
\begin{equation*}
\E_t \left[e^{\int_t^{t'}\frac{\I\xi}{U_s}\,\D s}\Big|U_{t'}\right]=\frac{I_{\sqrt{\left(1+\frac{2\kappa}{\varepsilon^2} \right)^2-\frac{8\I\xi}{\varepsilon^2}}}\left(\frac{2}{C_t}\sqrt{A_tU_{t'} U_t}\right)}{I_{1+\frac{2\kappa}{\varepsilon^2}}\left(\frac{2}{C_t}\sqrt{A_tU_{t'} U_t}\right)}.
\end{equation*}
This immediately leads to \eqref{closed_conchf}.

\section{Proof of Corollary \ref{coro1}}\label{deriveh}
 Apparently, we have
  \[
    \E_t[e^{\I\omega X_{t'}+\I\omega I_{t'}}] = \int_0^\infty \check G(t,X_t,I_t,V_t;t',\omega,\eta,v')\,\D v'
    =e^{\I\omega X_t+\I\eta I_t}\int_0^\infty g(t,V_t;t',\omega,\eta,v')\,\D v'.
  \]
  Then, it suffices to compute
\begin{eqnarray*}
  h(t,v;t',\omega,\eta) &=& \int_0^\infty g(t,v;t',\omega,\eta,v')\,\D v'\\
  &=& e^{a(t'-t)}\frac{A_t}{C_t}\int_0^\infty e^{-\frac{u+A_tu'}{C_t}}\left(\frac{A_tu'}{u}\right)^{\frac{1}{2}+\frac{\tilde\kappa}{\varepsilon^2}}
  I_{2c}\left(\frac{2}{C_t}\sqrt{A_tuu'}\right)\,\D u'\\
  &=& \frac{e^{a(t'-t)-\frac{u}{C_t}}}{C_tu^{\frac{1}{2}+\frac{\tilde\kappa}{\varepsilon^2}}}\int_0^\infty e^{-\frac{z}{C_t}}
  z^{\frac{1}{2}+\frac{\tilde\kappa}{\varepsilon^2}}I_{2c}\left(\frac{2\sqrt{uz}}{C_t}\right)\,\D z,
\end{eqnarray*}
where $z=A_tu'$. Using the equation
\[
  \int_0^\infty e^{-st}t^\iota I_\varsigma(\lambda\sqrt{t})\,\D t = \frac{\Gamma(\phi)}{\Gamma(\psi)}\frac{X^{\varsigma/2}}{s^{1+\iota}}M(\phi,\psi,X),
\]
where $\phi=1+\iota+\varsigma/2$, $\psi=1+\varsigma$, $X=\frac{\lambda^2}{4s}$ and $\mathfrak{R}(\phi,s)>0$, we obtain
\begin{eqnarray*}
  h(t,v;t',\omega,\eta) &=& e^{a(t'-t)-\frac{u}{C_t}}\frac{\Gamma(1-\alpha)}{\Gamma(2c+1)}\left(\frac{u}{C_t}\right)^{\tilde\alpha}M\left(1-\alpha,2c+1,\frac{u}{C_t}\right)\\
  &=& e^{a(t'-t)}\frac{\Gamma(\tilde\beta-\tilde\alpha)}{\Gamma(\tilde\beta)}\left(\frac{u}{C_t}\right)^{\tilde\alpha}M\left(\tilde\alpha,\tilde\beta,-\frac{u}{C_t}\right)\\
  &=&
  e^{a(t'-t)}\frac{\Gamma(\tilde\beta-\tilde\alpha)}{\Gamma(\tilde\beta)} \left(\frac{1}{C_tv}\right)^{\tilde\alpha}M\left(\tilde\alpha,\tilde\beta,-\frac{1}{C_tv}\right).
\end{eqnarray*}
Here, \[
\tilde\alpha=-\frac{1}{2}-\frac{\tilde\kappa}{\varepsilon^2}+c,~~~\text{and}~~~\quad\tilde\beta = 1+2c.\]
Note that the second equality follows from the identity:
\[
  M(a,b,z) = e^{z}M(b-a,b,-z).
\]

\end{appendices}

\end{document}